\documentclass[12pt]{article}

\usepackage{graphicx}
\usepackage[utf8]{inputenc}
\usepackage{a4wide}
\usepackage{amsmath,amsfonts,amssymb,amsthm}
\usepackage{xcolor}
\usepackage{graphicx}
\usepackage{xfrac}
\usepackage{tikz}
\usetikzlibrary{shapes.geometric, arrows,positioning}
\tikzstyle{regiao} = [rectangle, rounded corners, minimum width=2cm, minimum height=1cm,text centered, draw=black, fill=black!5]

\newtheorem{theorem}{Theorem}
\newtheorem{proposition}{Proposition}
\newtheorem{lemma}{Lemma}
\newtheorem{corollary}{Corollary}

\newcommand{\Rnaught}{\ensuremath{\mathcal{R}_0}}
\newcommand{\bydef}{:=}
\newcommand{\R}{\mathbb{R}}

\newcommand{\bv}{\mathbf{v}}
\newcommand{\mOp}{{-,0,+}}

\newcommand{\df}{{\mathrm{df}}}

\begin{document}

\title{Migrations, vaccinations and epidemic control%
}

\author{Fabio A. C. C. Chalub\thanks{ Departamento de Matem\'atica and Centro de Matem\'atica e Aplica\c c\~oes, Faculdade de Ci\^encias e Tecnologia, Universidade Nova de Lisboa, Quinta da Torre, 2829-516, Caparica, Portugal.}         \and
        Tiago J. Costa\thanks{ Funda\c c\~ao Champalimaud, Av. Bras\'ilia, 1400-038 Lisboa}				\and
        Paula Patr\'icio${}^\ast,$\thanks{Corresponding author: pcpr@fct.unl.pt} 
}


%


\date{\today}

\maketitle

\begin{abstract}
We consider three regions with different public health conditions. In the absence of migration among these regions, the first two have good health conditions and the disease free state is stable; for the third region, on the other hand, the only stable state is the endemic one. When migration is included in the model, we assume that the second region has a disease risk that makes its inhabitants prone to accept to be vaccinated, while the population in the first region tends to reject the vaccination, considered riskier that the disease. Therefore, the second region is a ``buffer zone'' between the two extremal regions. We study the basic reproductive ratio as a function of the vaccination in all regions and migration among them. This problem is studied numerically, showing explicit situations in which migration will have an overall positive effect in the disease dynamics, with and without vaccinations. We also find explicit formula in the limit of small (``closed borders'') and high migration (``open borders'').
\end{abstract}

\textbf{keywords:}
{Epidemic models, Vaccination, Meta-population, SIR model }

\section{Introduction}
\label{intro}
Compartmental models, SIR model being one of the most important examples, consider no other differences between individuals beyond disease status. Despite their simplicity, these models have provided important insight into disease transmission and control, e.g. the possibility to eradicate a disease without vaccinating the entire population, a concept know as ``herd immunity''. However, as no structure is introduced in the population, there is no information on who, where or when to vaccinate. To answer these questions, models structured by age, sex, risk-groups, coinfection, contact patterns, seasonality and/or spatial distribution, among other relevant characteristics, have been proposed. See textbooks and references works, e.g.,~\cite{MurrayI,Britton2003,Thieme_2003} for these and many other examples.

In this work we have turned our focus to space heterogeneity and its coupling with rational human behaviour with respect to vaccination, in a deterministic SIR model. Recently the introduction of space heterogeneity in compartmental models has been gaining increased recognition \cite{Youetal2013,Matrajtetal2013,Santermansetal2016,Rileyetal2015}.

Many recent works have also included human behaviour into epidemiological models. Some of these models assume that individuals choose to be vaccinated or not depending on the perceived risk of the vaccine relative to the risk of disease. Previous works show that assuming rational behaviour (i.e., risk minimization) of the population it is impossible to eliminate the disease for constant transmission rate \cite{Bauch2004,dOnofrio2007,Manfredi2009} or for  a time-periodic transmission function~\cite{Doutor2016}. To the best of our knowledge, however, no model in the literature considers spatial heterogeneity along with individual behaviour. 

We consider a patch model in order to accommodate space heterogeneity and human behaviour. Patch models with explicit movement can be tailored to mimic diffusive movement of individuals \cite{Cosner2015} and  are easier to analyse and suitable to investigate the implementation of control measures for different regions \cite{Wang&Mulone2003,Allenetal2007,Allenetal2009,Lachianyetal2012}.  Successful application to malaria~\cite{Xiao&Zou2014} and HIV \cite{Isdoryetal2015} have been conducted.

In this work, no explicit dynamics in the human behaviour is considered. However, the human behaviour in each patch is rational in the sense that the entire population accepts to be vaccinated or not according to the comparison of the perceived risk both of the disease and of the vaccine, which is assumed to be highly effective. Therefore a population accepts to be vaccinated only if the risk of getting the disease in a given patch is above a certain threshold, and does not accept otherwise.

We consider three neighbouring patches, corresponding to different transmission intensities, a first and second regions for which health conditions are good, i.e., such that in the absence of migration the disease free state is stable, and a third region where there is sustainable local transmission of the disease, represented by a stable endemic state in the SIR model. When migration is introduced in the model, the population in the second region will accept to be vaccinated, due to the proximity to the endemic region, but the population in the first region will still reject being vaccinated due to the rarity of local transmission.

The model is summarized by saying that there is local sustained transmission and vaccination refusal in the first region; no local sustained transmission, but vaccine acceptance in the second region, and finally, sustained transmission and vaccine acceptance in the third region.

We use this framework to compute the basic reproductive ratio. Despite being possible to obtain the analytical expression of $\mathcal{R}_0$, this expression is far too intricate to allow a simple analysis. Hence, we use large and small diffusion approximations of $\mathcal{R}_0$, to have insight in the parameters region that allow disease eradication. We also perform numerical simulations and we derive conditions to achieve disease elimination.

Our main goal is to find parameters allowing disease elimination by vaccinating only  areas for which the local population is prone to accept vaccination, given the inter-region migration rate.
As a secondary goal, we also want to understand the role of migration in the disease dispersion.

The paper is organized as follow: we finish the introduction revising SIR model with constant coefficients, in order to fix notation. In section~\ref{sec:model} we study the full model and prove the existence of the basic reproduction ratio $\mathcal{R}_0$ and  the existence of a unique disease endemic equilibrium for $\mathcal{R}_0>1$, which is locally asymptotically stable. In section~\ref{sec:control} we discuss how to jointly use vaccination and open borders to control an epidemic. In particular, we consider the level curves of $\mathcal{R}_0$ as a function of the vaccination in the second and third regions and of migration between all regions.
Finally, we discuss our findings in section~\ref{sec:discussion}.

%

%
\subsection{SIR models with vaccination}

In order to fix notation and terminology, we start by recovering some known facts from spatially homogeneous S(usceptible)-I(nfectious)-R(ecovered) model. Let $S$, $I$ and $R$ denote the total number of susceptible, infectious and recovered individuals respectively, and let the real constant parameters $\beta>0$, $\mu>0$ and $\gamma>0$ be the transmission rate, death/birth rate and recovery rate, respectively. The so called SIR model is given by $S'=-\beta SI+\mu(N-S)$, $I'=\beta SI-(\mu+\gamma)I$, $R'=\gamma I-\mu R$. Due to the fact that the total population $S+I+R=N$ is constant, it is customary to ignore one of the equations, in our case, the equation for $R$. The so called basic reproductive rate $\Rnaught\bydef \frac{\beta N}{\gamma+\mu}$,  is such that if $\Rnaught\le 1$ the disease free state (i.e., the equilibrium state $(S_0,I_0)$, with $I_0\equiv 0$) is asymptotically stable and if $\Rnaught>1$ the epidemic state (i.e, the equilibrium state $(S_1,I_1)$ with $I_1>0$) is the only stable state. 

Adult vaccination is included subtracting a $vS$ term in the first equation and adding a correspondent term in the last one. Therefore, we introduce
\[
 \mathcal{R}_0[v]=\frac{1}{N}S_0[v]\mathcal{R}_0[0]\ ,
\]
where $\mathcal{R}_0[0]=\frac{\beta N}{\gamma+\mu}$, the basic reproductive number for a totally non vaccinated population 
and $S_0[v]=\frac{\mu N}{\mu+v}$, the number of susceptible at equilibrium in the disease free state. The coefficient $\mathcal{R}_0[v]$ is
such that the disease free state is asymptotically stable if and only if $\mathcal{R}_0[v]\le 1$, otherwise the epidemic state is the only stable state. 
The minimum vaccination able to guarantee long term stability of disease-free state is given by $v_{\mathrm{min}}\bydef\max\{0,\mu(\mathcal{R}_0-1)\}$.

\section{Metapopulation model}
\label{sec:model}

We use the notation ``-'', ``0'' and ``+'' to identify the three regions described in the Introduction. Regions are classified according to different transmission rates and disease/vaccination risk: the first region, ``$-$'', has no local sustained transmission and vaccination refusal; the second region, ``$0$'', has no local sustained transmission, but vaccine acceptance; and third region, ``$+$'', we have sustained transmission and vaccine acceptance; see Fig.~\ref{fig:three_parts_ODE}.
We define $\lambda_-$, $\lambda_0$ and $\lambda_+$, to denote the relative size, which is proportional the number of individuals living in each region in the stationary state.
The model can be written as a system of $3\times 3$ ordinary differential equations (ODE), corresponding to one SIR model for each region. Let 
$S_{\mOp},I_{\mOp},R_{\mOp}$ be the fraction of susceptible, infectious and recovered individuals in regions ``-'', ``0'' and ``+'', respectively. The rate of vaccination in each region is given by $v_{\mOp}$ and the transmission rate
is given by $\beta_{\mOp}$. We consider two different  diffusion rates $d_-$ and $d_+$, between regions ``-'' and ``0''; and regions ``0'' and``+'', respectively.

The model can be written as the following system of ODEs
\begin{align}
\label{ODE_S-}
  S_-'&=\mu\lambda_- N-\beta_-I_-S_--(v_-+\mu)S_-+d_-(\lambda_-S_0-\lambda_0S_-)\ ,\\
\label{ODE_I-}
  I_-'&=\beta_-I_-S_--(\gamma+\mu)I_-+d_-(\lambda_-I_0-\lambda_0I_-)\ ,\\
\label{ODE_R-}
R_-'&=\gamma I_--\mu R_-+v_-S_-+d_-(\lambda_-R_0-\lambda_0R_-)\ ,\\
\label{ODE_S0}
S_0'&=\mu\lambda_0N-\beta_0I_0S_0-(v_0+\mu)S_0+d_-\lambda_0S_--(d_-\lambda_-+d_+\lambda_+)S_0+d_+\lambda_0S_+\ ,\\
\label{ODE_I0}
I_0'&=\beta_0I_0S_0-(\gamma+\mu)I_0+d_-\lambda_0I_--(d_-\lambda_-+d_+\lambda_+)I_0+d_+\lambda_0I_+\ ,\\
\label{ODE_R0}
R_0'&=\gamma I_0-\mu R_0+v_0S_0+d_-\lambda_0R_--(d_-\lambda_-+d_+\lambda_+)R_0+d_+\lambda_0R_+\ ,\\
\label{ODE_S+}
S_+'&=\mu\lambda_+ N-\beta_+I_+S_+-(v_++\mu)S_++d_+(\lambda_+S_0-\lambda_0S_+)\ ,\\
\label{ODE_I+}
I_+'&=\beta_+I_+S_+-(\gamma+\mu)I_++d_+(\lambda_+I_0-\lambda_0I_+)\ ,\\
\label{ODE_R+}
R_+'&=\gamma I_+-\mu R_++v_+S_++d_+(\lambda_+R_0-\lambda_0R_+)\ .
 \end{align}

 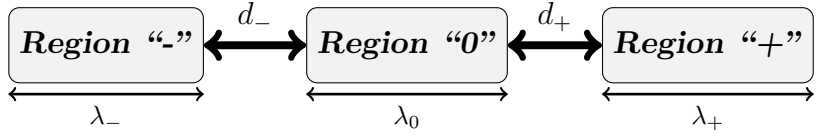
\begin{figure}
 \centering
  \begin{tikzpicture}
   \node (reg_-) [regiao] {\textbf{\textsl{Region ``-''}}};
\node (reg_0) [regiao,right of=reg_-,xshift=3cm] {\textbf{\textsl{Region ``0''}}};
\node (reg_+) [regiao,right of=reg_0,xshift=3cm] {\textbf{\textsl{Region ``+''}}};

\draw[<->,line width=3pt] (reg_-) -- node[above] {$d_-$} (reg_0);
\draw[<->,line width=3pt] (reg_0) -- node[above] {$d_+$}(reg_+);

\draw[<->,line width=1pt] ([yshift=-0.65cm]reg_-.west) -- node[below,font=\footnotesize] {$\lambda_-$} ([yshift=-0.65cm]reg_-.east);
\draw[<->,line width=1pt] ([yshift=-0.65cm]reg_0.west) -- node[below,font=\footnotesize] {$\lambda_0$} ([yshift=-0.65cm]reg_0.east);
\draw[<->,line width=1pt] ([yshift=-0.65cm]reg_+.west) -- node[below,font=\footnotesize] {$\lambda_+$} ([yshift=-0.65cm]reg_+.east);

  \end{tikzpicture}
\caption{Schematic representation of the three region model classified according to different transmission rates and disease/vaccination risk. It is assumed that, the first region "$-$" has no local sustained transmission and vaccination refusal; the second region "$0$" has no local sustained transmission, but vaccine acceptance; and finally, in the third region "$+$", we have sustained transmission and vaccine acceptance. The relative size of the three regions are denoted by $\lambda_\mOp$, respectively, and $d_\pm$ indicates the migration rate between adjacent sites.}
\label{fig:three_parts_ODE}
\end{figure}

We define $N_{+,0,-}=S_{+,-,0}+I_{+,0,-}+R_{+,0,-}$.  Total population $N(t)=N_-(t)+N_0(t)+N_+(t)$ is constant. For each region we have
\begin{align*}
 N_-'&=\mu(\lambda_- N-N_-)+d_-(\lambda_-N_0-\lambda_0N_-)\ ,\\
 N_0'&=\mu(\lambda_0 N-N_0)+d_-\lambda_0N_--(d_-\lambda_-+d_+\lambda_+)N_0+d_+\lambda_0N_+\ ,\\
 N_+'&=\mu(\lambda_+ N-N_+)+d_+(\lambda_+N_0-\lambda_0N_+)\ .
\end{align*}
Hence, if we assume that the total population is initially uniformly distributed, i.e. $N_{-,0,+}(0)=\lambda_{-,0,+} N$, then 
$N_\mOp(t)=\lambda_\mOp N$ for all $t\in\R_+$.

Furthermore, we assume that the transmission rates, $\beta_\mOp$, are such such that the basic reproductive ratio, in the absence of diffusion and vaccination, verifies $\mathcal{R}_0^-<\mathcal{R}_0^0<1<\mathcal{R}_0^+$, where $\mathcal{R}_0^\mOp=\frac{\beta_{\mOp}S_{\mOp}}{\gamma+\mu}=\frac{\beta_{\mOp}\lambda_{\mOp}N}{\gamma+\mu}$, i.e., without migration, regions ``-'' and  ``0'' would have no sustained transmission and in region ``+'' disease would be endemic.

\subsection{Model analysis} 

Consider the disease-free solution of system without vaccination, i.e., equations~(\ref{ODE_S-})--(\ref{ODE_R+}) with $I_\mOp\equiv0$ and $v_\mOp\equiv0$. In this case:

\begin{align*}
 0&=\mu\lambda_- N-\mu S_-^\df+d_-(\lambda_-S_0^\df-\lambda_0S_-^\df)\ ,\\
  0&=\mu\lambda_0 N-\mu S_0^\df+d_-\lambda_0S_-^\df-(d_-\lambda_-+d_+\lambda_+)S_0^\df+d_+\lambda_0S_+^\df \ ,\\
 0&=\mu\lambda_+ N-\mu S_+^\df+d_+(\lambda_+S_0^\df-\lambda_0S_+^\df)\ .
\end{align*}
It is clear that at the disease-free equilibrium without vaccination, $S_{\mOp}^{\df}=\lambda_{\mOp}N$.\\

To obtain the disease-free equilibrium with vaccination we rewrite equation~(\ref{ODE_S-}),~(\ref{ODE_S0}) and~(\ref{ODE_S+}) with $I_{\mOp}\equiv 0$, i.e.,
\begin{align*}
 (v_-+\mu+d_-\lambda_0) S_- -d_-\lambda_-S_0&=\mu\lambda_-N\ ,\\
  -d_-\lambda_0S_- + (v_0+\mu+d_-\lambda_-+d_+\lambda_+)S_0 -d_+\lambda_0S_+&=\mu\lambda_0 N\ ,\\
 -d_+\lambda_+S_0+(v_++\mu+d_+\lambda_0) S_+&=\mu\lambda_+ N\ .
\end{align*}
that can be written, in the matrix form, as the linear system $D\hat S=\mu N\Lambda$, $\hat S=(S_-,S_0,S_+)$, $\Lambda=(\lambda_-,\lambda_0,\lambda_+)$ and
 
 \[
 D=\left(\begin{matrix}
                           v_-+\mu+d_-\lambda_0&-d_-\lambda_- & 0\\
                           -d_-\lambda_0 & v_0 + \mu + d_-\lambda_-+ d_+\lambda_+&  - d_+\lambda_0\\
                           0 & -d_+\lambda_+ & v_+ + \mu + d_+\lambda_0
                         \end{matrix}\right)\ .
\]

From explicit calculations, we find that $D^{-1}$ is entrywise positive, and therefore $\hat S=\mu N D^{-1}\Lambda$ is uniquely defined and entrywise positive. Finally, we conclude that,

\begin{proposition}
System~(\ref{ODE_S-})--(\ref{ODE_R+}) has a unique disease free equilibrium.
\end{proposition}

The exact solution will shed no light on this particular problem so we omit it. Latter, we will present the expression up to the first order in the diffusion coefficients.

Following \cite{ShuaiDriessche2013}, the basic reproduction number of system ~(\ref{ODE_S-})--(\ref{ODE_R+}) is given by the spectral radius of the next generation matrix $FV^{-1}$,
\[
\mathcal{R}_0[\bv]=\rho(FV^{-1}),\qquad\bv\bydef(v_-,v_0,v_+)\ ,
\]
where 
\begin{align*}
F&=\left(\begin{matrix}
\beta_-S_-^{\df}& 0 &0\\
0 &\beta_0S_0^{\df}& 0\\
0& 0 &\beta_-S_+^{\df}
\end{matrix}\right)\quad \text{and} \\
V&=\left(\begin{matrix}
\gamma+\mu+d_-\lambda_0& -d_-\lambda_- &0\\
-d_-\lambda_0 &\gamma+\mu+d_-\lambda_-+d_+\lambda_+& -d_+\lambda_0\\
0& -d_+\lambda_+ &\gamma+\mu+d_+\lambda_0
\end{matrix}\right)\ .
\end{align*}

Let
\[
\Gamma=\left\lbrace (S_{\mOp},I_{\mOp})\in\R^{3\times 2}:S_{\mOp}+I_{\mOp}\leq \lambda_{\mOp} \mbox{ and } S_{\mOp}\leq S_{\mOp}^{\df} \right\rbrace .
\]
Note that $\Gamma$ invariant under the flow of system~(\ref{ODE_S-})--(\ref{ODE_R+}). Let $\mathring{\Gamma}$ denote the interior of $\Gamma$ and $\partial \Gamma$ its boundary.
We finish this section with the following stability result:

\begin{theorem}\label{thm:sharp_threshold_prop}
Consider the  model~(\ref{ODE_S-})--(\ref{ODE_R+}). For $d_-$ and $d_+>0$, there exist $\mathcal{R}_0[\bv]$ such that
\begin{enumerate}
\item if $\mathcal{R}_0[\bv]\leq 1$, then the disease free equilibrium is global asymptotically stable, and
\item if $\mathcal{R}_0[\bv]>1$, then the disease free equilibrium is unstable and there is a unique endemic equilibrium that is locally asymptotically stable in $\mathring{\Gamma}$.
\end{enumerate}

\end{theorem}
\begin{proof}
We use the notation of~\cite{LiShuai2009}, and define $\Lambda_i=\mu \lambda_{\mOp}N$, $d_i^S=v_{\mOp}+\mu$ and $d_i^I=\mu$ and matrices $A=B=\displaystyle \left(\begin{smallmatrix}
0&d_-\lambda_-&0\\
d_-\lambda_0&0&d_+\lambda_0\\
0&d_+\lambda_+&0
\end{smallmatrix}\right)$, which are irreducible for $d_-$ and $d_+>0$.From~\cite[Proposition 3.2]{LiShuai2009}, we conclude the existence of a locally asymptotically stable endemic equilibrium $P^*=(S^*_{\mOp},I^*_{\mOp})\in\mathring{\Gamma}$. Stability results follow from  \cite[Theorem 4.1(3)]{LiShuai2009}.
\end{proof}

\section{``Open borders'', vaccination and epidemic control}
\label{sec:control}

Structural changes in the epidemic dynamics, as e.g., the use of other prophylactic measures that results in decreased transmission rate $\beta$, may be implemented but are not the main concern of the present work, in which we focus only in vaccination and migration. In this section, we will study the parameter region that prevents an outbreak (i.e., able to make the disease-free state stable). We assume that diffusion between regions is equal $d=d_+=d_-$.

Using a computer-algebra software (namely, Maple 13), we obtain asymptotic expressions for the basic reproductive ratio in the limit $d\to0$ and compare with the limit $d\to\infty$ of the basic reproductive ratio, in order to have an heuristics in the behaviour of $\mathcal{R}_0$ as a function of migration. 

\begin{lemma}
 Consider the disease free-state of the  system~(\ref{ODE_S-})--(\ref{ODE_R+}) and let $S_\mOp^\df[\bv]$ be the solution with vaccination given by $\bv=(v_-,v_0,v_+)$. Assume furthermore that in the limit of null migration $\lim_{d \to 0} \mathcal{R}_0[\bv] \geq 1$.. Then, up to first order in $d\ll 1$, we find
 \begin{equation*}
  \mathcal{R}_0[\bv]
  =\frac{\lambda_+ N\beta_+\mu}{(\gamma+\mu)(\mu+v_+)}\left[1-d\lambda_0\left(\frac{1}{\gamma+\mu}+\frac{1}{v_++\mu}-\frac{1}{v_0+\mu}\right)\right]+\mathcal{O}(d^2)\ .
 \end{equation*}
In the limit of large migration, we find
\begin{equation*}
\lim_{d\to\infty} \mathcal{R}_0[\bv]=\frac{\mu N(\beta_-\lambda_-^2+\beta_0\lambda_0^2+\beta_+\lambda_+^2)}{(\gamma+\mu)[(\mu+v_-)\lambda_-+(\mu+v_0)\lambda_0+(\mu+v_+)\lambda_+]}\ .
\end{equation*}

\end{lemma}

\begin{proof}
 For the limit of small $d$, we compute the three eigenvalues of the matrix $FV^{-1}$. In the limit $d\to 0$, these are $\mathcal{R}_0^-[v_-]$, $\mathcal{R}_0^0[v_0]$, and $\mathcal{R}_0^+[v_+]$. We assume that for $d\ll 1$, the spectral ratio of $FV^{-1}$ will be given by the eigenvalue representing an order $d$ perturbation of $\mathcal{R}_0^+[v_+]\bydef\frac{\mu\lambda_+ N\beta_+}{(\gamma+\mu)(v_+\mu)}$, from which the first result follows. For the limit $d\to\infty$, we find that 2 eigenvalues of $FV^{-1}$ converge to 0, and therefore, the spectral radius is given by the limit of the third eigenvalue.
\end{proof}

As a direct consequence of the previous result, we have a situation for which a small migration $d\ll 1$ will decrease the value of $\mathcal{R}_0$, and for some parameter combination, simply opening the borders might eliminate the epidemic.
Of course, $\mathcal{R}_0$ gives information of the global behaviour of the system, i.e., $\mathcal{R}_0$ does not give information about the dynamics of a single patch, which means that even if the global situation becomes better with a larger $d$ (i.e., a smaller $\mathcal{R}_0$), if $\mathcal{R}_0$  is still larger than one, individuals in the left region, and possibly in the central one, may be refractory to increase $d$. 

\begin{corollary}
 If
 \[
  \frac{1}{\gamma+\mu}+\frac{1}{v_++\mu}-\frac{1}{v_0+\mu}>0\ ,
 \]
then, the introduction of a small motility between regions will decrease the overall basic reproductive ratio. In particular, in the absence of vaccination, the overall effect of a small, but positive, diffusion, is to decrease the basic reproductive ratio. 
\end{corollary}

A second important corollary shows that the effect of large migration may be positive, depending on the specific parameters of the problem.

\begin{corollary}
 Assume
 \begin{equation}\label{eq:cor2}
 \left(\frac{\beta_0\lambda_0}{\beta_+\lambda_+}-\frac{\mu+v_0}{\mu+v_+}\right)\lambda_0+\left(\frac{\beta_-\lambda_-}{\beta_+\lambda_+}-\frac{\mu+v_-}{\mu+v_+}\right)\lambda_-<0\ .
 \end{equation}
Then there is $d_{\mathrm{c}}$ large enough such that for $d>d_{\mathrm{c}}$, $\displaystyle\mathcal{R}_0[\bv]<\lim_{d\to0}\mathcal{R}_0[\bv]$. In particular, equation~(\ref{eq:cor2}) is true if
\[
 \max\left\{\frac{\beta_0\lambda_0}{\mu+v_0},\frac{\beta_-\lambda_-}{\mu+v_-}\right\}<\frac{\beta_+\lambda_+}{\mu+v_+}\ .
\]
\end{corollary}

If we assume a positive migration $d>0$,  sustained transmission in the third region makes it possible to get disease from a migrant, therefore, the risk in each region will depend on the disease dynamics in the adjacent regions. In particular, we will consider a scenario in which the risk in the middle region is such that the population accepts to be vaccinated (as is the case in the third region), while the population in the first region, due to real or imaginary perceived vaccination risks, do no accept to be vaccinated. Note that migrants in the first region come exclusively from the second region, where there is no sustained transmission. Therefore, the second region is a ``buffer zone''. 

Figure~\ref{fig:R0equal1} shows the curves $\mathcal{R}_0[\bv]=1$, with $\bv=(0,v_0,v_+)$, for different values of $d$ and for two different parameter sets. Note that in the limit $d\to0$, the vaccination in the second region, $v_0$, becomes irrelevant for the dynamics. For $d>0$, and $v_+>0$, but still insufficient to prevent the outbreak, it may be possible (depending on $v_+$) to increase only $v_0$ such that $\mathcal{R}_0[\bv]<1$, i.e., an outbreak can be prevented with extra vaccination only on the ``buffer zone''. For large $d$ the curves $\mathcal{R}_0[\bv]=1$ converge to straight lines, with inclination depending on the model parameters, representing the fact that an increase in $v_0$ and $v_+$ have a linear effect in the overall dynamics. In Fig.~\ref{fig:R0equal1} (above) the diffusion alone is not able to prevent the outbreak, while in the example below, this is the case.

\begin{figure}
 \includegraphics[width=0.45\textwidth]{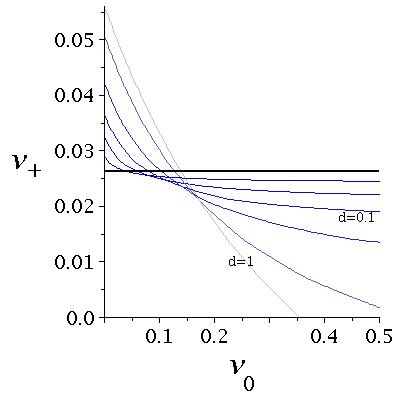}
 \includegraphics[width=0.45\textwidth]{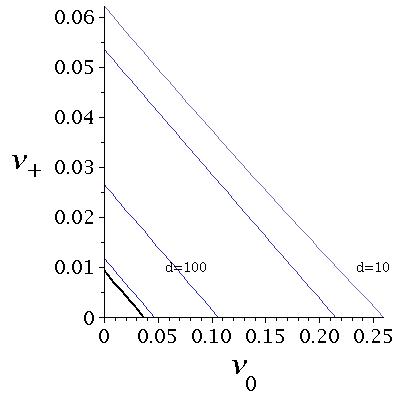}\\
 \includegraphics[width=0.45\textwidth]{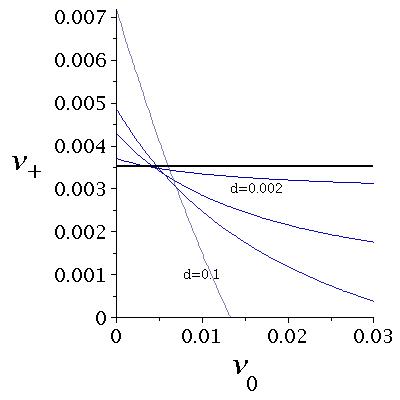}
 \includegraphics[width=0.45\textwidth]{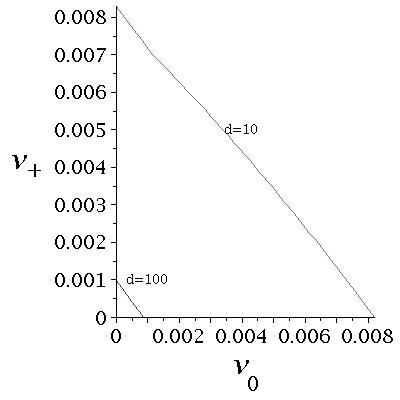}\\
 \caption{$\mathcal{R}_0[\bv]=1$ with $\bv=(0,v_0,v_+)$ as function of $v_0$ ($x$ axis) and $v_+$ ($y$ axis), for different values of $d$. The stable disease free region is above the corresponding line. We consider in all examples $\gamma=52$, $\mu=1/80$. In the left pictures we show the behaviour for small diffusion (the thick line indicates the limit $d\to0$), while in the right pictures, we show the behaviour for large diffusion; in this case, the thick line indicates the limit $d\to\infty$. Above we divide the domain in $\hat\lambda\bydef(\lambda_-,\lambda_0,\lambda_+)=\left(\frac{1}{2},\frac{1}{10},\frac{2}{5}\right)$ and $\hat\beta\bydef(\beta_-,\beta_0,\beta_+)=(30,75,400)$, implying in $\hat{\mathcal{R}}_0\bydef(\mathcal{R}_0^-,\mathcal{R}_0^0,\mathcal{R}_0^+)\approx(0.10,0.14,3.08)$, in the absence of vaccination and diffusion. In particular, in the absence of vaccination $\mathcal{R}_0|_{d=0}=3.08$ and $\mathcal{R}_0 |_{d\to\infty}=1.28$. The lines indicates diffusion equal to 1, 0.5, 0.25, 0.1, 0.05, 0.02 (left) and 10, 100, 1000, 10000 (right). Below: $\hat\lambda=\left(\frac{1}{3},\frac{1}{3},\frac{1}{3}\right)$, $\hat\beta=(30,150,200)$, implying in $\hat{\mathcal{R}}_0\approx(0.19,0.96,1.28)$. It is clear, in this case, that the required vaccination to extinguish the disease will be much smaller. In fact, for $\bv=0$, $\lim_{d\to\infty}\mathcal{R}_0<1$~and therefore the increased diffusion will be enough to make the disease free state stable, even in the absence of vaccination. We plot the critical line $\mathcal{R}_0=1$ for $d$ equals to 0.1, 0.02, 0.01, 0.002 (left) and 10, 100 (right).} 
\label{fig:R0equal1}
 \end{figure}

\section{Discussion}
\label{sec:discussion}

This work is concerned with the study of epidemiological dynamics in metapopulation, in which the environment is naturally divided in three regions: in the first and  second regions, there is no sustained disease transmission. However the existence of migrations between the second region and the third one (with, say, poor sanitary conditions) may bring sporadic cases of the disease to the second and, as  consequence, the migration between the first and the second may still bring rare cases to first region. Therefore, assuming that the population in the first region is not willing to be vaccinated, we study different scenarios, involving migration between sites and  vaccinations in the second and third regions such that outbreaks can be prevented, bringing the basic reproductive ratio $\mathcal{R}_0$ below 1.

In a distinct interpretation of the same model, the second region can work as a ``buffer zone''. We study critical vaccination coverages in the the second and third regions, that are able to prevent disease spread, without vaccinating the population in the first region.

There are some clear limitations in our model. The first one is that a large number of migrants (specially in the limit $d\to\infty$) can change the sanitary conditions of the destination, therefore changing the value of the transmission rates $\beta_\mOp$, that we assumed constant. Furthermore, in all analytical and numerical studied we assume the same migration rate between both regions and that the migration rate does not depend on the disease status. 

With minor modifications, all analysis presented here could be changed to consider $d_+\ne d_-$. However, in section~\ref{sec:control}, we opt for $d_+=d_-$, because this made both the analytical and the numerical analysis simpler.

Even though our model does not include any dynamics -- rational or not -- for the human behaviour, we hope it can provide some insights on the importance of including this kind of structure in epidemiology systems.
Finally, we did not study vaccination costs. This, however, can be made without difficulty, minimizing the vaccination cost in the region above the curve $\mathcal{R}_0=1$, given $d$, in figure~\ref{fig:R0equal1}.

%

\section*{Acknowledgements}
This work was partially supported by FCT/Portugal project  EXPL/MAT-CAL/0794/2013, by Strategic Project UID/MAT/00297/2013 (Centro de Matem\'atica e Aplica\c c\~oes, Universidade Nova de Lisboa) and PhD scholarship PD/BD/128299/2017 (T. Costa, F. Champalimaud). We also thank Max Souza (UFF, Brazil) for several suggestions.


%

%

\end{document}